\documentclass[lettersize,journal]{IEEEtran}
\usepackage{amsmath,amsfonts}
\usepackage{amssymb}  
\usepackage{amsthm}
\usepackage{algorithmic}
\usepackage{algorithm}
\usepackage{array}
\usepackage[caption=false,font=normalsize,labelfont=sf,textfont=sf]{subfig}
\usepackage{textcomp}
\usepackage{stfloats}
\usepackage{url}
\usepackage{verbatim}
\usepackage{graphicx}

\usepackage{cite}
\newtheorem{theorem}{Theorem}       
\newtheorem{lemma}{Lemma}           
  
\newtheorem{assumption}{Assumption}

\begin{document}

\title{An Evolutionary Computation Framework for Multi-Agent Q-Learning with Mean-Field Environmental Feedback}

\author{Lichen Wang, Shijia Hua, Linjie Liu
\thanks{L. Wang, S. Hua, L. Liu  are with College of Science, Northwest A \& F University, Yangling, 712100, China.}
\thanks{L. Wang is also with the School of Mathematical Sciences, University of Electronic Science and Technology of China, Chengdu 611731, China.}
\thanks{Corresponding authors: Linjie Liu (email: linjieliu1992@nwafu.edu.cn).}}



\maketitle

\begin{abstract}
	Multi-agent reinforcement learning in networked populations is governed by the interaction between individual adaptation, local encounters, and changing environmental conditions. To study this interaction, we formulate a coupled learning--environment model in which agents update stateless $Q$-values on a fixed graph, while their population-average behavior drives an environmental variable that dynamically modifies the payoff matrix. Under a first-order mean-field closure, we derive a deterministic transport equation for the population distribution of $Q$-values and couple it with a projected discrete update for the environmental state. The resulting model is evaluated against finite-network Monte Carlo simulations on random regular, Erd\H{o}s--R\'enyi, Barab\'asi--Albert, and random geometric graphs. Across the tested parameter ranges, the mean-field system reproduces the main macroscopic cooperation and environmental trajectories, and the trajectory-level root-mean-square error generally decreases with population size and average degree. The analysis further shows that environmental feedback reshapes the learned action-value ordering, while reinforcing feedback can produce pronounced dependence on the initial learning bias and resource level. The environmental timescale also plays an important role: a rapid response can drive the resource state to a boundary before learning adapts, whereas a slower response preserves the interaction between behavioral learning and environmental recovery. These results provide a population-level description of coupled reinforcement learning and environmental dynamics and characterize the performance of the mean-field approximation within the tested network and parameter ranges.
\end{abstract}


\section{Introduction}

Multi-agent reinforcement learning (MARL) studies how multiple adaptive decision makers learn while their rewards and observations depend on the behavior of other simultaneously adapting agents~\cite{pi2025dynamic,zeng2025evolutionary}. Its value-based foundation includes $Q$-learning, which updates action values from experienced rewards, together with exploration rules that translate those values into stochastic policies~\cite{watkins1992technical,kaelbling1996reinforcement}. The multi-agent setting introduces non-stationarity, coordination, and scalability challenges that do not arise in the same form for an isolated learner, and it has consequently been studied from the perspectives of multi-agent systems, reinforcement learning, evolutionary learning, and deep MARL~\cite{stone2000multiagent,busoniu2008comprehensive,zhang2021selective,hernandezleal2019survey,gronauer2021multiagent,zhang2017fmrq,zhang2021collaborative}. When interaction and information exchange are constrained by a graph, an agent responds to a local neighborhood rather than to a well-mixed population, so the collective learning process can depend on connectivity, degree heterogeneity, and local sampling~\cite{zhang2018networked}. These features motivate a population-level theory that remains connected to the microscopic learning rule while making explicit which network effects are retained and which are removed by approximation.

Evolutionary game theory provides a complementary language for population adaptation. Its classical foundations characterize evolutionarily stable strategies and the dynamics induced by payoff differences~\cite{wei2024moral,zhu2023equilibrium,zhu2022networked,zhu2022nash,smith1973logic,smith1982evolution,taylor1978evolutionary,hofbauer2003evolutionary,nowak2004evolutionary}. Subsequent work identified general mechanisms supporting cooperation and clarified why temporal and spatial structure can lead to behavior beyond a well-mixed replicator description~\cite{nowak2006five,roca2009beyond}. A substantial literature has also connected multi-agent learning to evolutionary dynamics: reinforcement and related adaptive rules can generate selection--mutation or replicator-like equations under suitable approximations~\cite{tuyls2005evolutionary,tuyls2007evolutionary,sato2003coupled}. Nevertheless, a frequency-based replicator equation and a value-based reinforcement learner represent different state variables. Replicator dynamics directly evolve strategy shares, whereas $Q$-learning retains action-specific estimates of accumulated reward and uses exploration to convert those estimates into behavior.

The interaction graph adds another level of complexity. Graph topology can change the conditions under which cooperation emerges, and heterogeneous degree, update asymmetry, and structured mixing can alter both transient and long-run outcomes~\cite{ohtsuki2006simple,santos2006evolutionary,nowak2009evolutionary,szabo2007evolutionary,perc2017statistical}. Studies on social diversity, asymmetric interaction and replacement processes, and interdependent networks show that structural features can alter evolutionary outcomes beyond what is predicted by homogeneous population averages~\cite{santos2008social,ohtsuki2007breaking,wang2013interdependent,ren2024reputation,wei2025moral,hou2024multiagent}. Extensions to multiplayer, edge-heterogeneous, and higher-order interactions reinforce the same conclusion~\cite{pena2016multiplayer,su2019edgediversity,battiston2020beyond}. A first-order mean-field closure provides a tractable population-level description by replacing local interaction statistics with aggregate quantities. Because this approximation omits local correlations, we evaluate its accuracy across population sizes, neighborhood sizes, and network topologies in Section III.

This population-level treatment becomes especially important when collective behavior changes the environment that shapes subsequent rewards. In social--ecological and common-resource systems, cooperative behavior may improve a shared environmental state, while depletion or recovery can change the relative incentives for subsequent actions. Evolutionary game models with environmental feedback have shown that this bidirectional coupling can generate oscillations, resource collapse, distinct long-run outcomes, and limit cycles~\cite{weitz2016oscillating,tilman2020evolutionary,chen2018punishment,gong2022limit,gavin2026learning,liang2025coevolution}. Recent work has further classified broad families of two-strategy games with environmental feedback and studied asymmetric behavioral effects on the environment~\cite{ito2024complete,shao2019asymmetrical}. These studies establish the importance of feedback and timescale separation, while the broader eco-evolutionary literature emphasizes that environmental change and behavioral adaptation can recursively reshape one another~\cite{govaert2019feedbacks}. Most analytical models in this line evolve strategy frequencies through selection, imitation, or prescribed switching. They do not explicitly track how individual learners retain reward history in action values or how softmax action selection translates those values into behavior that subsequently affects environmental change.

For population-level analyses of multi-agent $Q$-learning, Hu et al. derived a Fokker--Planck equation for the evolution of $Q$-value distributions in well-mixed populations playing repeated symmetric games~\cite{hu2019modelling}. They subsequently developed a continuity-equation model for $Q$-learning dynamics in population games~\cite{hu2022dynamics}. For graph-structured populations, Chu et al. formulated a model for regular graphs, while Liu et al. extended this framework to a broader range of network structures, including random and scale-free graphs~\cite{chu2022formal,liu2025formal}. Yuan et al. further investigated $Q$-learning dynamics in networked stochastic games using a pair-approximation approach~\cite{yuan2026dynamics}. Together, these studies provide the methodological foundation for population-level descriptions of graph-based $Q$-learning.
Zhang et al. incorporated environmental feedback into networked multi-agent reinforcement learning to promote cooperation in evolutionary games~\cite{zhang2023environmental}. Their study focuses on the learning-based promotion of cooperation under environmental feedback, whereas the distribution-level models above characterize $Q$-value evolution under exogenous payoff structures. These two directions leave open a population-level formulation in which the distribution of action values and an endogenous environmental state evolve self-consistently.

To address this gap, we formulate a coupled learning--environment model in which agents interact on a fixed graph, select actions through a softmax policy, and update stateless $Q$-values from neighbor-averaged rewards, while the population-average strategy drives a global environmental variable that modifies subsequent payoffs. Under explicitly stated neighbor-independence and population-averaged selection closures, we construct a deterministic first-order transport equation for the population distribution of action values and couple it self-consistently to a projected discrete environmental update. The theoretical formulation allows a finite action set and general environment-dependent payoff and feedback maps; the numerical experiments use a binary-action benchmark, affine payoff interpolation, and linear environmental feedback to isolate interpretable mechanisms. We compare the transport prediction with finite-network Monte Carlo simulations on random regular, Erd\H{o}s--R\'enyi, Barab\'asi--Albert, and random geometric graphs, and quantify the approximation error across population sizes and average degrees.

The main contributions are summarized as follows:
\begin{enumerate}
	\item We formulate a bidirectionally coupled model linking microscopic stateless $Q$-learning on a graph, the population distribution of action values, and an endogenous environmental state.
	
	\item Under explicitly stated mean-field closure assumptions, we derive a deterministic first-order transport equation for the $Q$-value distribution and couple it self-consistently to a projected environmental update defined on the common discrete learning timescale.
	
	\item We quantitatively evaluate the population-level prediction against finite-network simulations across four representative topologies and investigate how environmental response time, feedback strength, policy selectivity, and initial action-value bias shape cooperation and resource dynamics.
\end{enumerate}

The remainder of this paper is organized as follows. Section II defines the networked learning--environment model and derives the coupled transport--environment system. Section III evaluates the approximation against finite-network simulations and examines payoff configurations, finite-size error, environmental response time, policy selectivity, feedback strength, and initial-condition dependence. Section IV summarizes the findings, limitations, and possible extensions to nonlinear feedback, state-dependent learning, and adaptive networks.

\section{Models and Methods}

\subsection{Multi-Agent Systems on Graphs}

Let $G=(\mathcal{N},\mathcal{E})$ be a finite undirected graph, where
$\mathcal{N}=\{1,2,\ldots,N\}$ is the agent set and
$\mathcal{E}\subseteq\mathcal{N}\times\mathcal{N}$ is the edge set. For each
agent $i\in\mathcal{N}$, define its neighborhood by
\begin{equation*}
	\mathcal{N}_i := \{\,j\in\mathcal{N}:(i,j)\in\mathcal{E}\,\},
\end{equation*}
and denote its degree by
\begin{equation*}
	k_i:=|\mathcal{N}_i|.
\end{equation*}
Because the interaction payoff introduced below is averaged over an agent's
neighbors, we consider graphs without isolated nodes, such that $k_i\geq 1$
for every $i\in\mathcal{N}$.

The empirical degree distribution of $G$ is defined as
\begin{equation*}
	\rho(k) := \frac{1}{N} \big|\{\,i\in\mathcal{N}:k_i=k\,\}\big|, \qquad k\in\{1,\ldots,N-1\}.
\end{equation*}
It satisfies
\begin{equation*}
	\rho(k)\geq 0, \qquad \sum_{k=1}^{N-1}\rho(k)=1.
\end{equation*}

All agents select actions from the common finite action set
\begin{equation*}
	\mathcal{A} = \{a_1,a_2,\ldots,a_m\},
\end{equation*}
where $m\in\mathbb{N}$ denotes the number of available actions.

\subsection{Stateless $Q$-Learning and Softmax Policy}

For each agent $i\in\mathcal{N}$ at discrete time $t$, let
\begin{equation*}
	\mathbf{q}_i(t) = \big( q_i(a_1,t), q_i(a_2,t), \ldots, q_i(a_m,t) \big)^\top \in\mathbb{R}^m
\end{equation*}
denote its action-value vector, where $q_i(a_p,t)$ is the value assigned by
agent $i$ to action $a_p$. The learning process is stateless in the sense
that each agent maintains one value for each action and does not condition
these values on an explicitly observed environmental state.

Given an inverse-temperature parameter $\beta>0$, define the softmax policy
$\pi_\beta:\mathbb{R}^m\to\Delta^{m-1}$ by
\begin{equation*}
	\pi_\beta(a_p\mid\mathbf{q}) := \frac{\exp\big(\beta q(a_p)\big)} {\displaystyle\sum_{\ell=1}^{m} \exp\big(\beta q(a_\ell)\big)}, \qquad p=1,\ldots,m.
\end{equation*}
Let $A_i(t)\in\mathcal{A}$ denote the random action selected by agent $i$
at time $t$. Its action-selection probability is
\begin{equation*}
	x_i(a_p,t) := \pi_\beta\big(a_p\mid\mathbf{q}_i(t)\big).
\end{equation*}
A smaller value of $\beta$ produces a more exploratory policy, whereas a
larger value assigns greater probability to actions with higher estimated
values.

Let $\alpha\in(0,1]$ denote the learning rate. After selecting $A_i(t)$
and receiving reward $r_i(t)$, agent $i$ updates its action values
according to
\begin{equation*}
	q_i(a,t+1) = q_i(a,t) + \alpha\, \mathbb{I}_{\{A_i(t)=a\}} \big( r_i(t)-q_i(a,t) \big), \qquad a\in\mathcal{A},
\end{equation*}
where $\mathbb{I}_{\{\cdot\}}$ is the indicator function. Thus, only the
value of the selected action is updated, whereas the values of the unselected
actions remain unchanged.

\subsection{Environmental Dynamics}

Let $n(t)\in[0,1]$ denote a global environmental state shared by all
agents, where $n=0$ and $n=1$ represent two limiting environmental
conditions, such as resource depletion and resource abundance, respectively.
Let
\begin{equation*}
	R(n)\in\mathbb{R}^{m\times m}
\end{equation*}
denote the payoff matrix at environmental state $n\in[0,1]$. At each time
$t$, this matrix is shared by all agents and interaction edges.

For actions $a_p,a_q\in\mathcal{A}$, let $\mathbf{e}_p$ denote the
$p$-th standard basis vector in $\mathbb{R}^m$. The payoff received by an
agent selecting $a_p$ against an agent selecting $a_q$ is
\begin{equation*}
	r(a_p,a_q;n) := \mathbf{e}_p^\top R(n)\mathbf{e}_q.
\end{equation*}
At time $t$, agent $i$ interacts with all agents in its neighborhood and
receives the neighbor-averaged payoff
\begin{equation*}
	r_i(t) = \frac{1}{k_i} \sum_{j\in\mathcal{N}_i} r\big( A_i(t),A_j(t);n(t) \big).
\end{equation*}

The population-average mixed strategy is defined as
\begin{equation*}
	\bar{\mathbf{x}}(t) = \big( \bar{x}(a_1,t), \ldots, \bar{x}(a_m,t) \big)^\top \in\Delta^{m-1},
\end{equation*}
where
\begin{equation*}
	\bar{x}(a_p,t) := \frac{1}{N} \sum_{i=1}^{N} x_i(a_p,t), \qquad p=1,\ldots,m.
\end{equation*}

The environmental state evolves according to
\begin{equation}
	\dot{n}(t) = n(t)\big(1-n(t)\big) f\big(\bar{\mathbf{x}}(t)\big),
	\label{eq:environment-dynamics}
\end{equation}
where $f:\Delta^{m-1}\to\mathbb{R}$ is the environmental feedback function
that specifies how the population-average strategy affects environmental
change.

The finite-network implementation of the coupled learning--environment dynamics
is summarized in Algorithm~\ref{alg:finite-network-q-learning}. At each iteration $t$, all
agents compute their mixed strategies from the current action values
$\mathbf{q}_i(t)$ and sample their actions accordingly. The resulting payoffs are
then used to update the action values, and the environmental state is advanced
using the current population-average strategy $\bar{\mathbf{x}}(t)$. Define $\Pi_{[0,1]}(z):=\min\{1,\max\{0,z\}\}$.

\begin{algorithm}[t]
	\caption{Finite-Network Monte Carlo Simulation of Stateless $Q$-Learning with Environmental Feedback}
	\label{alg:finite-network-q-learning}
	\begin{algorithmic}[1]
		\REQUIRE Graph $G=(\mathcal{N},\mathcal{E})$; payoff map $R(n)$;
		feedback function $f$; learning rate $\alpha$; inverse temperature
		$\beta$; environmental timescale $\epsilon$; horizon $T$
		\STATE Initialize $q_i(a,0)$ for all $i\in\mathcal{N}$ and $a\in\mathcal{A}$
		\STATE Initialize the environmental state $n(0)\in[0,1]$
		\FOR{$t=0$ to $T-1$}
		\FOR{each agent $i\in\mathcal{N}$}
		\STATE Compute $x_i(a,t)\gets\pi_\beta(a\mid\mathbf{q}_i(t))$ for all $a\in\mathcal{A}$
		\STATE Sample $A_i(t)\sim x_i(\cdot,t)$
		\ENDFOR
		\STATE Compute $\bar{\mathbf{x}}(t)\gets|\mathcal{N}|^{-1}
		\sum_{i\in\mathcal{N}}\mathbf{x}_i(t)$
		\STATE Evaluate the payoff matrix $R(n(t))$
		\FOR{each agent $i\in\mathcal{N}$}
		\IF{$k_i>0$}
		\STATE Compute $r_i(t)\gets\dfrac{1}{k_i}\sum_{j\in\mathcal{N}_i}
		[R(n(t))]_{A_i(t),A_j(t)}$
		\ELSE
		\STATE Set $r_i(t)\gets0$
		\ENDIF
		\STATE Update
		\begin{equation*}
			q_i(A_i(t),t+1)\gets q_i(A_i(t),t) +\alpha\bigl(r_i(t)-q_i(A_i(t),t)\bigr)
		\end{equation*}
		\STATE Set $q_i(a,t+1)\gets q_i(a,t)$ for all $a\neq A_i(t)$
		\ENDFOR
		\STATE Update
		\begin{equation*}
			n(t+1)\gets\Pi_{[0,1]} \left[n(t)+\frac{1}{\epsilon}n(t)(1-n(t)) f\bigl(\bar{\mathbf{x}}(t)\bigr)\right]
		\end{equation*}
		\ENDFOR
	\end{algorithmic}
\end{algorithm}

\subsection{Neighbor-Action Configuration}

For each agent $i\in\mathcal{N}$, action $a_q\in\mathcal{A}$, and time
$t$, define
\begin{equation*}
	c_{i,q}(t) := \big| \{\,j\in\mathcal{N}_i:A_j(t)=a_q\,\} \big|, \qquad q=1,\ldots,m,
\end{equation*}
as the number of neighbors of agent $i$ that select action $a_q$. The
neighbor-action configuration associated with agent $i$ is the random
vector
\begin{equation*}
	\boldsymbol{\gamma}_i(t) := \big( c_{i,1}(t), \ldots, c_{i,m}(t) \big)^\top\in \mathbb{N}^m, \quad \sum_{q=1}^{m}c_{i,q}(t)=k_i.
\end{equation*}

For a node of degree $k$, define the configuration space
\begin{equation*}
	\Gamma(k) := \left\{ \boldsymbol{\eta} = (\eta_1,\ldots,\eta_m)^\top \in \mathbb{N}^m: \sum_{q=1}^{m}\eta_q=k \right\}.
\end{equation*}
Accordingly, the neighbor-action configuration of agent $i$ satisfies
\begin{equation*}
	\boldsymbol{\gamma}_i(t)\in\Gamma(k_i).
\end{equation*}

\subsection{Payoff Conditional on Neighbor Configuration}

Consider an agent of degree $k$ and a neighbor-action configuration
$\boldsymbol{\eta}\in\Gamma(k)$. The neighbor-averaged payoff from selecting
action $a_p$ at environmental state $n$ is defined as
\begin{equation*}
	u_p(\boldsymbol{\eta};n) := \frac{1}{k} \sum_{q=1}^{m} \eta_q\,r(a_p,a_q;n) = \frac{1}{k} \mathbf{e}_p^\top R(n) \boldsymbol{\eta}.
\end{equation*}
Here, $\eta_q$ is the number of neighbors selecting action $a_q$, and the
factor $1/k$ converts the accumulated interaction payoff into a
neighbor-averaged payoff.

For agent $i$, if $A_i(t)=a_p$, then the reward used in its $Q$-value
update is
\begin{equation*}
	r_i(t) = u_p\big( \boldsymbol{\gamma}_i(t); n(t) \big).
\end{equation*}

\subsection{Mean-Field Approximation}
An agent's realized payoff depends on the actions selected by its neighbors, and the resulting neighbor-action configuration is both time varying and random. An exact population-level analysis must therefore track the joint learning dynamics and local action correlations across the network, which is generally intractable for large interacting populations. To obtain a closed description of the macroscopic dynamics, we adopt a mean-field approximation that replaces local action-selection probabilities with the corresponding population averages. This approximation suppresses local correlations and degree-conditioned fluctuations; its predictive accuracy on finite networks is assessed through agent-based simulations in Section III. The specific closure assumption is stated below.

\begin{assumption}[Mean-Field Approximation]
	\label{ass:mean-field-closure}
	At each time $t$, conditional on the population-average mixed strategy
	$\bar{\mathbf{x}}(t)$, the random actions selected by the neighbors of any
	agent $i\in\mathcal{N}$ are mutually independent. Moreover, every neighbor
	$j\in\mathcal{N}_i$ is approximated as selecting action
	$a_q\in\mathcal{A}$ with probability
	\begin{equation*}
		\mathbb{P}\big( A_j(t)=a_q \mid \bar{\mathbf{x}}(t) \big) = \bar{x}(a_q,t), \qquad q=1,\ldots,m.
	\end{equation*}
\end{assumption}

\begin{lemma}[Multinomial Distribution of Neighbor-Action Configurations]
	\label{lem:neighbor-configuration}
	Under Assumption~\ref{ass:mean-field-closure}, fix a time $t$ and regard the
	population-average mixed strategy $\bar{\mathbf{x}}(t)$ as given. For an
	agent $i$ with degree $k_i=k$, the neighbor-action configuration satisfies
	\begin{equation*}
		\boldsymbol{\gamma}_i(t) \sim \operatorname{Multinomial} \big( k,\bar{\mathbf{x}}(t) \big).
	\end{equation*}
	Equivalently, for every
	$\boldsymbol{\eta}=(\eta_1,\ldots,\eta_m)^\top\in\Gamma(k)$,
	\begin{equation*}
		\mathbb{P}\big( \boldsymbol{\gamma}_i(t)=\boldsymbol{\eta} \mid k_i=k \big) = \frac{k!}{\displaystyle\prod_{p=1}^{m}\eta_p!} \prod_{p=1}^{m} \bar{x}(a_p,t)^{\eta_p}.
	\end{equation*}
\end{lemma}

\begin{proof}
	Fix an agent $i$ with $k_i=k$. Under Assumption~\ref{ass:mean-field-closure}, the
	actions selected by the $k$ neighbors of agent $i$ are conditionally
	independent given $\bar{\mathbf{x}}(t)$, and each neighbor selects action
	$a_p$ with probability $\bar{x}(a_p,t)$. Therefore, the vector that counts
	the numbers of neighbors selecting the $m$ available actions follows a
	multinomial distribution with $k$ trials and category-probability vector
	$\bar{\mathbf{x}}(t)$. Hence, for every
	$\boldsymbol{\eta}\in\Gamma(k)$,
	\begin{equation*}
		\mathbb{P}\big( \boldsymbol{\gamma}_i(t)=\boldsymbol{\eta} \mid k_i=k \big) = \frac{k!}{\displaystyle\prod_{p=1}^{m}\eta_p!} \prod_{p=1}^{m} \bar{x}(a_p,t)^{\eta_p}.
	\end{equation*}
	This proves the lemma.
\end{proof}

For agent $i$, action $a_p\in\mathcal{A}$, and time $t$, define the
increment in the corresponding action value by
\begin{equation*}
	\Delta q_i(a_p,t) := q_i(a_p,t+1)-q_i(a_p,t).
\end{equation*}
The stateless $Q$-learning update rule then gives
\begin{equation*}
	\Delta q_i(a_p,t) = \alpha\, \mathbb{I}_{\{A_i(t)=a_p\}} \big( r_i(t)-q_i(a_p,t) \big).
\end{equation*}
Thus, the value of action $a_p$ changes only when agent $i$ realizes that
action at time $t$.

\begin{lemma}[Mean $Q$-Value Increment under the Mean-Field Approximation]
	\label{lem:individual-mean-increment}
	Under Assumption~\ref{ass:mean-field-closure}, for each agent $i$ and time $t$, the mean
	increment in the value of action $a_p$ is
	\begin{equation}
		\mathbb{E}\big[ \Delta q_i(a_p,t) \big] = \alpha\,x_i(a_p,t) \left( \mu_p\big(k_i,t;n(t)\big) - q_i(a_p,t) \right),
		\label{eq:individual-mean-increment}
	\end{equation}
	where the expectation is taken over the action selected by agent $i$ and
	the neighbor-action configuration at time $t$, with
	$\mathbf{q}_i(t)$, $k_i$, $n(t)$, and
	$\bar{\mathbf{x}}(t)$ treated as given. The mean-field expected payoff from
	selecting action $a_p$ for an agent of degree $k$ is
	\begin{equation*}
		\mu_p(k,t;n) := \sum_{\boldsymbol{\eta}\in\Gamma(k)} \mathbb{P}\big( \boldsymbol{\gamma}_i(t)=\boldsymbol{\eta} \mid k_i=k \big) u_p(\boldsymbol{\eta};n).
	\end{equation*}
\end{lemma}

\begin{proof}
	From the increment form of the stateless $Q$-learning update,
	\begin{equation*}
		\Delta q_i(a_p,t) = \alpha\, \mathbb{I}_{\{A_i(t)=a_p\}} \big( r_i(t)-q_i(a_p,t) \big).
	\end{equation*}
	At time $t$, the value $q_i(a_p,t)$ is fixed. Moreover,
	\begin{equation*}
		\mathbb{P}\big(A_i(t)=a_p\big) = x_i(a_p,t),
	\end{equation*}
	and, conditional on $A_i(t)=a_p$, the mean-field expected reward is
	\begin{equation*}
		\mathbb{E}\big[ r_i(t) \mid A_i(t)=a_p \big] = \mu_p\big(k_i,t;n(t)\big).
	\end{equation*}
	Taking the expectation of the increment and combining these two relations
	gives
	\begin{equation*}
		\mathbb{E}\big[ \Delta q_i(a_p,t) \big] = \alpha\,x_i(a_p,t) \left( \mu_p\big(k_i,t;n(t)\big) - q_i(a_p,t) \right),
	\end{equation*}
	which proves the lemma.
\end{proof}

Lemma~\ref{lem:individual-mean-increment} gives the expected payoff conditional on the
agent degree. Under the common payoff matrix and neighbor-averaged payoff
adopted in the present model, the multinomial mean-field approximation yields
\begin{equation*}
	\mu_p(k,t;n(t)) = \mathbf e_p^\top R(n(t))\bar{\mathbf x}(t),
\end{equation*}
so the explicit dependence on $k$ cancels in the numerical setting considered
here. We nevertheless retain the degree-conditioned notation
$\mu_p(k,t;n(t))$ and the averaging over $\rho(k)$, because this formulation
also accommodates extensions with degree-dependent payoff matrices or payoff
aggregation rules, for which the expected payoff generally depends on $k$.
We next introduce a representative-agent description to construct the
population-level mean drift.

\begin{lemma}[Population-Averaged Mean $Q$-Value Increment]
	\label{lem:population-mean-increment}
	Under Assumption~\ref{ass:mean-field-closure}, let $q(a_p,t)$ denote the value assigned to
	action $a_p$ by a representative agent. We adopt a population-averaged
	selection closure in which the representative agent selects action $a_p$
	with probability $\bar{x}(a_p,t)$. Averaging the degree-conditioned payoff
	over the empirical degree distribution $\rho(k)$, the mean increment is
	approximated by
	\begin{equation}
		\mathbb{E}\big[ \Delta q(a_p,t) \big] = \alpha\,\bar{x}(a_p,t) \left( \sum_k \rho(k)\, \mu_p\big(k,t;n(t)\big) - q(a_p,t) \right).
		\label{eq:population-mean-increment}
	\end{equation}
	Here,
	\begin{equation*}
		\sum_k \rho(k)\, \mu_p\big(k,t;n(t)\big)
	\end{equation*}
	is the expected payoff from action $a_p$, averaged over the degree of a
	uniformly sampled agent.
\end{lemma}

\begin{proof}
	Let $K$ denote the degree of an agent sampled uniformly from the population.
	By the definition of the empirical degree distribution,
	\begin{equation*}
		\mathbb{P}(K=k)=\rho(k).
	\end{equation*}
	Conditional on $K=k$, Lemma~\ref{lem:individual-mean-increment}, together with the
	population-averaged selection closure, gives
	\begin{equation*}
		\mathbb{E}\big[ \Delta q(a_p,t) \mid K=k \big] = \alpha\,\bar{x}(a_p,t) \left( \mu_p\big(k,t;n(t)\big) - q(a_p,t) \right).
	\end{equation*}
	Applying the law of total expectation over $K$ yields
	\begin{equation*}
		\begin{aligned} \mathbb{E}\big[ \Delta q(a_p,t) \big] &= \sum_k \rho(k)\, \mathbb{E}\big[ \Delta q(a_p,t) \mid K=k \big] \\
		&= \alpha\,\bar{x}(a_p,t) \sum_k \rho(k) \left( \mu_p\big(k,t;n(t)\big) - q(a_p,t) \right) \\
		&= \alpha\,\bar{x}(a_p,t) \left( \sum_k \rho(k)\, \mu_p\big(k,t;n(t)\big) - q(a_p,t) \right), \end{aligned}
	\end{equation*}
	where the last equality uses $\sum_k\rho(k)=1$. This proves the lemma.
\end{proof}

We identify the expected $Q$-value increment per learning round with a
continuous-time drift. This mean-drift interpolation retains the conditional
mean of the discrete update while neglecting its stochastic fluctuations.
\begin{theorem}[Mean-Field Transport Equation]
	\label{thm:mean-field-transport}
	Under Assumption~\ref{ass:mean-field-closure} and the population-averaged selection closure,
	let $p(\mathbf{q},t)$ denote the population density of the $Q$-value vector
	\begin{equation*}
		\mathbf{q} = \big( q(a_1),\ldots,q(a_m) \big)^\top \in\mathbb{R}^m.
	\end{equation*}
	In the drift-based continuous-time approximation, the density evolves
	according to the transport equation
	\begin{equation}
		\frac{\partial}{\partial t}p(\mathbf{q},t) = -\sum_{p=1}^{m} \frac{\partial}{\partial q(a_p)} \Big( p(\mathbf{q},t)\,v_p(\mathbf{q},t) \Big),
		\label{eq:transport-equation}
	\end{equation}
	where the $p$-th component of the drift field is
	\begin{equation}
		v_p(\mathbf{q},t) = \alpha\,\bar{x}(a_p,t) \left( \sum_k \rho(k)\, \mu_p\big(k,t;n(t)\big) - q(a_p) \right).
		\label{eq:transport-drift}
	\end{equation}
	The population-average action-selection probability is obtained
	self-consistently from the density:
	\begin{equation}
		\bar{x}(a_p,t) = \int_{\mathbb{R}^m} x(a_p;\mathbf{q})\, p(\mathbf{q},t)\, \mathrm{d}\mathbf{q}, \qquad p=1,\ldots,m.
		\label{eq:population-strategy-density}
	\end{equation}
\end{theorem}

\begin{proof}
	Let $\mathbf{Q}(t)$ denote the $Q$-value vector of a representative agent
	in the first-order mean-field continuous-time approximation, whose evolution is
	governed by
	\begin{equation*}
		\frac{\mathrm{d}}{\mathrm{d}t}Q(a_p,t) = v_p\bigl(\mathbf{Q}(t),t\bigr), \qquad p=1,\ldots,m.
	\end{equation*}
	Let $\varphi\in C_c^\infty(\mathbb{R}^m)$ be an arbitrary smooth test
	function with compact support. By the chain rule,
	\begin{equation*}
		\frac{\mathrm{d}}{\mathrm{d}t} \varphi\big(\mathbf{Q}(t)\big) = \sum_{p=1}^{m} \frac{\partial\varphi}{\partial q(a_p)} \big(\mathbf{Q}(t)\big) v_p\big(\mathbf{Q}(t),t\big).
	\end{equation*}
	Taking expectations gives
	\begin{equation*}
		\frac{\mathrm{d}}{\mathrm{d}t} \mathbb{E}\big[ \varphi\big(\mathbf{Q}(t)\big) \big] = \mathbb{E}\left[ \sum_{p=1}^{m} \frac{\partial\varphi}{\partial q(a_p)} \big(\mathbf{Q}(t)\big) v_p\big(\mathbf{Q}(t),t\big) \right].
	\end{equation*}
	
	Let $p(\mathbf{q},t)$ denote the probability density of
	$\mathbf{Q}(t)$. The left-hand side can be written as
	\begin{equation*}
		\frac{\mathrm{d}}{\mathrm{d}t} \mathbb{E}\big[ \varphi\big(\mathbf{Q}(t)\big) \big] = \int_{\mathbb{R}^m} \varphi(\mathbf{q})\, \frac{\partial p(\mathbf{q},t)}{\partial t} \,\mathrm{d}\mathbf{q},
	\end{equation*}
	whereas the right-hand side becomes
	\begin{equation*}
		\int_{\mathbb{R}^m} \sum_{p=1}^{m} \frac{\partial\varphi(\mathbf{q})}{\partial q(a_p)} v_p(\mathbf{q},t) p(\mathbf{q},t) \,\mathrm{d}\mathbf{q}.
	\end{equation*}
	
	Because $\varphi$ has compact support and the probability flux vanishes at
	the boundary, integration by parts yields
	\begin{equation*}
		\begin{aligned} &\int_{\mathbb{R}^m} \sum_{p=1}^{m} \frac{\partial\varphi(\mathbf{q})}{\partial q(a_p)} v_p(\mathbf{q},t) p(\mathbf{q},t) \,\mathrm{d}\mathbf{q} \\
		&\quad = -\int_{\mathbb{R}^m} \varphi(\mathbf{q}) \sum_{p=1}^{m} \frac{\partial}{\partial q(a_p)} \left( p(\mathbf{q},t)v_p(\mathbf{q},t) \right) \,\mathrm{d}\mathbf{q}. \end{aligned}
	\end{equation*}
	
	Combining the preceding expressions, we obtain
	\begin{equation*}
		\int_{\mathbb{R}^m} \varphi(\mathbf{q}) \left[ \frac{\partial p(\mathbf{q},t)}{\partial t} + \sum_{p=1}^{m} \frac{\partial}{\partial q(a_p)} \left( p(\mathbf{q},t)v_p(\mathbf{q},t) \right) \right] \,\mathrm{d}\mathbf{q} = 0.
	\end{equation*}
	Since this identity holds for every
	$\varphi\in C_c^\infty(\mathbb{R}^m)$, the density satisfies, in the
	distributional sense,
	\begin{equation*}
		\frac{\partial p(\mathbf{q},t)}{\partial t} = -\sum_{p=1}^{m} \frac{\partial}{\partial q(a_p)} \left( p(\mathbf{q},t)v_p(\mathbf{q},t) \right).
	\end{equation*}
\end{proof}

\subsection{Coupled Mean-Field Learning and Environment Dynamics}

Theorem~\ref{thm:mean-field-transport} gives the within-round transport equation for the population density $p(\mathbf q,t)$. To couple this equation with the discrete learning process, let $r=0,1,2,\ldots$ denote the learning-round index and $s\in[0,1]$ the within-round time. Define $p_r(\mathbf q,s):=p(\mathbf q,r+s)$ and $n_r:=n(r)$. For $s\in(0,1)$ and $p=1,\ldots,m$, the coupled mean-field transport--environment system is
\begin{equation}
	\left\{ \begin{aligned} \frac{\partial p_r(\mathbf q,s)}{\partial s} &= -\sum_{p=1}^{m}\frac{\partial}{\partial q(a_p)}\left[p_r(\mathbf q,s)v_{p,r}(\mathbf q,s)\right], \\
	v_{p,r}(\mathbf q,s) &= \alpha\bar{x}_r(a_p,s)\left[\sum_k\rho(k)\mu_p(k,s;n_r)-q(a_p)\right], \\
	\bar{x}_r(a_p,s) &= \int_{\mathbb R^m}\pi_\beta(a_p\mid\mathbf q)p_r(\mathbf q,s)\,\mathrm d\mathbf q, \\
	p_{r+1}(\mathbf q,0) &= p_r(\mathbf q,1), \\
	n_{r+1} &= \Pi_{[0,1]}\left[n_r+\frac{1}{\epsilon}n_r(1-n_r)f\bigl(\bar{\mathbf x}_r(0)\bigr)\right]. \end{aligned} \right.
	\label{eq:coupled-transport-environment}
\end{equation}
Here, $\epsilon>0$ controls the environmental response timescale relative to the learning dynamics. During round $r$, the environmental state is fixed at $n_r$, while the density and population-average strategy evolve over $s\in[0,1]$. Both $p_{r+1}(\mathbf q,0)$ and $n_{r+1}$ are determined from the state at the beginning of the round, giving a synchronous learning--environment update.

The projected environmental update is the round-based realization of the environmental rate law in~\eqref{eq:environment-dynamics}. If the resource loss during one round exceeds the remaining resource level, the projection assigns $n_{r+1}=0$, representing complete resource depletion. Similarly, an update above one represents resource saturation. The boundaries $n=0$ and $n=1$ are absorbing under this update.

Fig.~\ref{fig:framework} summarizes the finite-network learning process and its population-level mean-field description for the binary-action case $\mathcal A=\{C,D\}$.

\begin{figure*}[t]
	\centering
	\includegraphics[width=\textwidth]{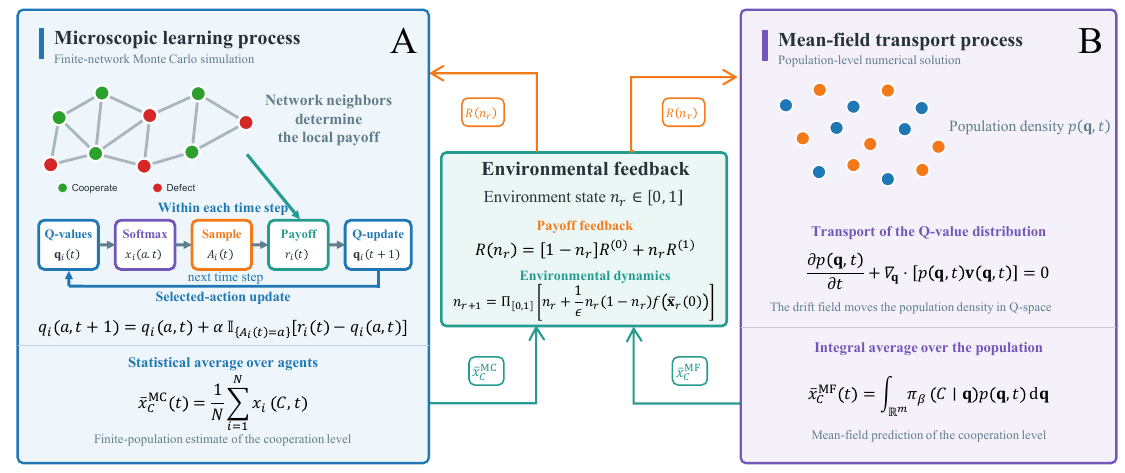}
	\caption{Micro-to-macro representation of the coupled learning--environment dynamics for the binary-action case $\mathcal A=\{C,D\}$. The finite-network Monte Carlo simulation and the mean-field transport equation follow the same synchronous learning-round convention and are coupled through the population-average strategy and the projected environmental update.}
	\label{fig:framework}
\end{figure*}

\section{Results}
To evaluate the general framework, we specialize it to a two-action game with $\mathcal{A}=\{C,D\}$, where $C$ and $D$ denote cooperation and defection, respectively. The limiting environmental states $n=0$ and $n=1$ are associated with the payoff matrices $R^{(0)}$ and $R^{(1)}$, respectively. With the rows and columns ordered as $(C,D)$, these matrices are given by
\begin{equation*}
	R^{(0)} = \begin{pmatrix} R_0 & S_0\\
	T_0 & P_0 \end{pmatrix}, \qquad R^{(1)} = \begin{pmatrix} R_1 & S_1\\
	T_1 & P_1 \end{pmatrix}.
\end{equation*}
For each $s\in\{0,1\}$, $R_s$ and $P_s$ denote the payoffs from mutual cooperation and mutual defection, respectively, whereas $S_s$ is the payoff received by a cooperating agent interacting with a defector, and $T_s$ is the payoff received by a defecting agent interacting with a cooperator.

For the numerical experiments, the payoff matrix is assumed to vary linearly
with the environmental state:
\begin{equation}
	R(n) = (1-n)R^{(0)} + nR^{(1)}, \qquad n\in[0,1].
	\label{eq:payoff-interpolation}
\end{equation}

We adopt the following environmental feedback function:
\begin{equation}
	f\bigl(\bar{\mathbf{x}}(t);\theta\bigr) = \theta\,\bar{x}(C,t)-\bar{x}(D,t) = \theta\,\bar{x}(C,t) - \bigl[1-\bar{x}(C,t)\bigr],
	\label{eq:environmental-feedback}
\end{equation}
where $\bar{x}(C,t)$ and $\bar{x}(D,t)$ denote the population-average
probabilities of cooperation and defection, respectively. The parameter
$\theta>0$ controls the relative contribution of cooperation to environmental
change.

For a given initial environmental state $n_0$, the initial payoff matrix is
obtained by evaluating~\eqref{eq:payoff-interpolation} at $n=n_0$.

The common support of the initial $Q$-value distributions is defined as
\begin{equation*}
	\begin{aligned} Q_{\min} &= \max_{a\in\{C,D\}} \min_{b\in\{C,D\}} \big[R(n_0)\big]_{a,b}, \\
	Q_{\max} &= \min_{a\in\{C,D\}} \max_{b\in\{C,D\}} \big[R(n_0)\big]_{a,b}. \end{aligned}
\end{equation*}
For the payoff configurations considered below, $Q_{\min}\leq Q_{\max}$. The initial values $q_i(C,0)$ and $q_i(D,0)$ are independently sampled from Beta distributions with shape parameters $(\eta_1,\eta_2)$ and $(\eta_2,\eta_1)$, respectively, and then linearly scaled from $[0,1]$ to $[Q_{\min},Q_{\max}]$. Varying $(\eta_1,\eta_2)$ produces different initial learning biases.

The agent-based simulation and the numerical solution of the mean-field transport equation use consistent synchronous updates. At each learning round, the learning dynamics are evaluated using the current environmental state, while the environmental state is advanced by the projected update in~\eqref{eq:coupled-transport-environment} based on the current population-average strategy. One outer time step corresponds to one learning round. The complete finite-network update procedure is summarized in Algorithm~\ref{alg:finite-network-q-learning}. Unless otherwise specified, the numerical experiments use an inverse temperature of $\beta=2$, a learning rate of $\alpha=0.4$, a population size of $N=1000$, a feedback strength of $\theta=2$, and an environmental timescale of $\epsilon=10$.

Although the analytical formulation is stated for $k_i\geq1$, a small fraction of isolated nodes may occur in finite ER and RGG realizations. For such nodes, we assign a zero realized payoff and apply the same $Q$-learning update.

To construct the environmental payoff configurations considered below, we introduce the following Harmony and Prisoner's Dilemma payoff matrices:
\begin{equation}
	R^{\mathrm{H}} = \begin{pmatrix} 5 & 1\\
	3 & 0 \end{pmatrix}, \qquad R^{\mathrm{PD}} = \begin{pmatrix} 3 & 0\\
	5 & 1 \end{pmatrix}.
	\label{eq:benchmark-payoffs}
\end{equation}
Under the action ordering $(C,D)$, cooperation strictly dominates defection in $R^{\mathrm{H}}$. In contrast, defection strictly dominates cooperation in $R^{\mathrm{PD}}$, although mutual cooperation yields a higher payoff than mutual defection.

Based on the two payoff matrices in \eqref{eq:benchmark-payoffs}, we consider four environmental payoff configurations. Env. 1 and Env. 2 serve as static baselines in which the payoff matrix is independent of the environmental state, whereas Env. 3 and Env. 4 introduce opposite directions of environment-dependent payoff feedback.

\begin{enumerate}
	\item \textbf{Env. 1: Static Harmony Game.}
	Setting $R^{(0)}=R^{(1)}=R^{\mathrm{H}}$ removes environment-dependent payoff variation and maintains cooperation as the strictly dominant action. This configuration provides a cooperative baseline for evaluating the learning dynamics.
	
	\item \textbf{Env. 2: Static Prisoner's Dilemma.}
	Setting $R^{(0)}=R^{(1)}=R^{\mathrm{PD}}$ also removes environment-dependent payoff variation but makes defection the strictly dominant action. This configuration provides the corresponding defection-favoring baseline.
	
	\item \textbf{Env. 3: Compensatory Payoff Feedback.}
	Setting $R^{(0)}=R^{\mathrm{H}}$ and $R^{(1)}=R^{\mathrm{PD}}$ causes the game to shift from a Harmony game toward a Prisoner's Dilemma as the environmental state improves from $n=0$ to $n=1$. Environmental improvement therefore weakens the payoff incentive for cooperation, producing a compensatory feedback loop.
	
	\item \textbf{Env. 4: Reinforcing Payoff Feedback.}
	Setting $R^{(0)}=R^{\mathrm{PD}}$ and $R^{(1)}=R^{\mathrm{H}}$ reverses the preceding relationship. As the environmental state improves, the payoff structure becomes increasingly favorable to cooperation, thereby reinforcing the behavioral changes that support environmental improvement.
\end{enumerate}

\begin{figure*}[t]
	\centering
	\includegraphics[width=\textwidth]{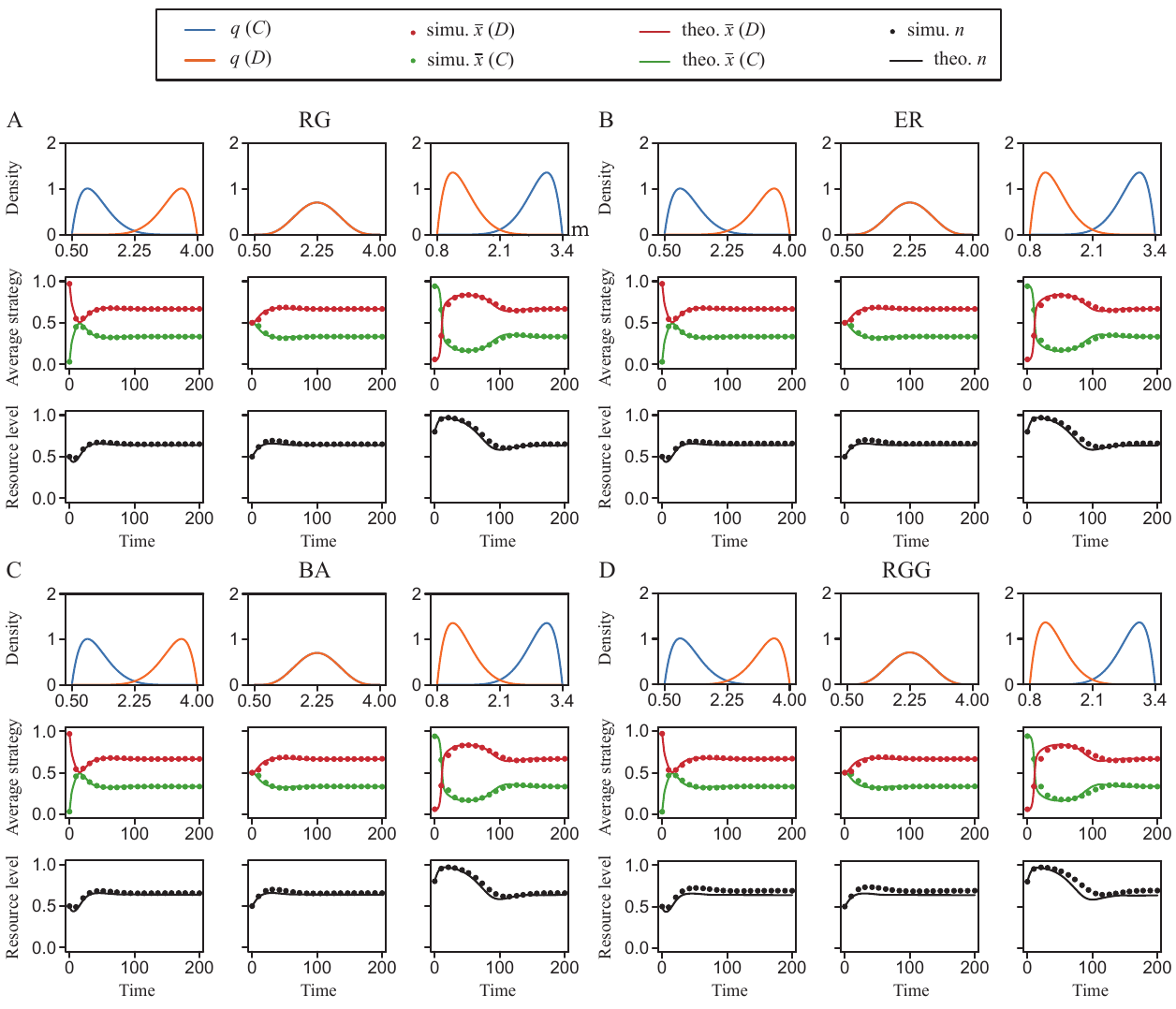}
	\caption{Comparison of agent-based simulations and mean-field transport predictions under Env. 3 across four network topologies. Panels A--D correspond to RG, ER, BA, and RGG networks, respectively. Within each panel, the three columns use $(\eta_1,\eta_2,n_0)=(2,8,0.5)$, $(5,5,0.5)$, and $(8,2,0.8)$, respectively. The first row shows the initial $Q$-value densities, the second row shows the population-average cooperation and defection probabilities, and the third row shows the environmental state $n(t)$. Markers denote averages over 100 independent agent-based runs, and solid lines denote the corresponding mean-field predictions.}
	\label{fig:network-comparison}
\end{figure*}

To assess the sensitivity of the mean-field approximation to interaction topology, we compare the agent-based and mean-field dynamics under Env. 3 on four representative networks, as shown in Fig.~\ref{fig:network-comparison}. The random regular graph (RG) provides a homogeneous-degree baseline, the Erd\H{o}s--R\'enyi (ER) network represents random connectivity, the Barab\'asi--Albert (BA) network introduces degree heterogeneity, and the random geometric graph (RGG) incorporates spatially localized interactions. The networks are generated with the same nominal mean degree $\langle k\rangle=4$, while the learning and environmental parameters are held fixed for each initial condition.

Under all three initial conditions, the cooperation probability, defection probability, and environmental state follow similar time courses on the four networks. Different initial $Q$-value distributions produce different changes during the early stage, but these differences gradually decrease, and the systems approach similar long-run levels. The mean-field curves generally follow the agent-based simulation results on the RG, ER, and BA networks. A more visible difference appears on the RGG network, particularly for the environmental state, which may result from the stronger local correlations induced by spatially constrained interactions. These results show that the mean-field model captures the overall evolution observed in the simulations, although it does not reproduce every finite-network trajectory exactly.

Because the tested networks display broadly similar population-level trends in Fig.~\ref{fig:network-comparison}, we use the RG network as a homogeneous baseline in the subsequent parameter studies where topology is not the primary variable. This choice helps isolate the effects of the learning and environmental parameters without introducing additional degree heterogeneity. The qualitative agreement observed in Fig.~\ref{fig:network-comparison} is quantified in Fig.~\ref{fig:rmse} by evaluating the trajectory-level RMSE over different population sizes and average degrees.

\begin{figure*}
	\centering
	\includegraphics[width=1\textwidth]{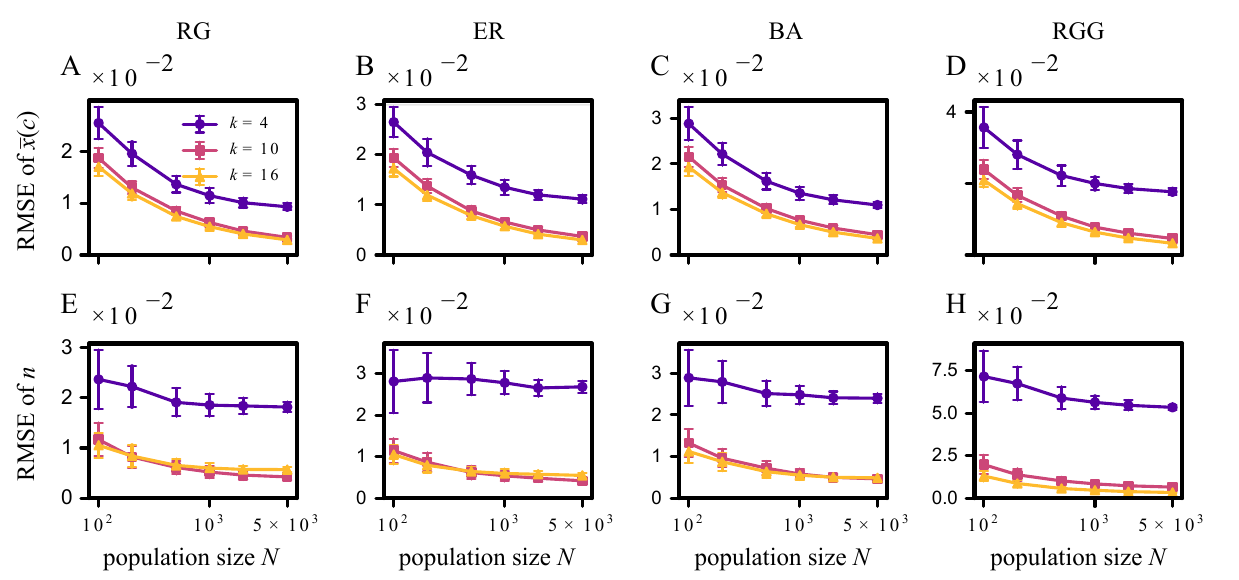}
	\caption{Trajectory-level root-mean-square errors between finite-network simulations and the mean-field transport prediction under Env.~3. Panels A--D show the RMSE of the population-average cooperation trajectory $\bar{x}(C,t)$, whereas Panels E--H show the RMSE of the environmental trajectory $n(t)$. The columns correspond to RG, ER, BA, and RGG networks, respectively. The horizontal axis denotes the population size $N$, and the curves correspond to average degrees $k=4$, $10$, and $16$. Each data point represents the mean RMSE over 100 simulations, and the error bars indicate the corresponding standard deviation.}
	\label{fig:rmse}
\end{figure*}

To quantify the discrepancy between the simulation and mean-field trajectories, we use the root-mean-square error (RMSE)
\begin{equation}
	\mathrm{RMSE}(y(t)) = \left[ \frac{1}{T+1} \sum_{t=0}^{T} \left( y_{\mathrm{sim}}(t)-y_{\mathrm{theo}}(t) \right)^2 \right]^{1/2},
\end{equation}
where $y(t)$ denotes either the population-average cooperation level $\bar{x}(C,t)$ or the environmental state $n(t)$.

We quantify the discrepancy between the finite-network simulations and the mean-field transport prediction using the trajectory-level RMSE defined above (Fig.~\ref{fig:rmse}). The upper row reports the RMSE of the population-average cooperation trajectory $\bar{x}(C,t)$, whereas the lower row reports the RMSE of the environmental trajectory $n(t)$. In both rows, the horizontal axis denotes the population size $N$, and the curves correspond to target mean degrees $\langle k\rangle=4$, $10$, and $16$. Each point is averaged over 100 simulations, and the error bars show the standard deviation of the RMSE across these simulations.

Across the four network topologies, both error measures generally decrease as $N$ increases. Increasing the average degree also reduces the discrepancy, with the largest errors consistently observed for $k=4$. This reduction is particularly clear for the cooperation trajectory, whereas the environmental RMSE exhibits small nonmonotonic variations for some network--degree combinations. Nevertheless, the RMSE remains on the order of $10^{-2}$ throughout the tested parameter range, and similar dependencies on $N$ and $k$ are observed for RG, ER, BA, and RGG networks. These results are consistent with finite-population fluctuations and limited neighborhood sampling contributing to the deviation from the mean-field prediction. The decreasing error with increasing $N$ and $k$ supports the use of the transport equation as an approximation of the macroscopic learning--environment dynamics within the tested network and parameter ranges.

\begin{figure*}[t]
	\centering
	\includegraphics[width=\textwidth]{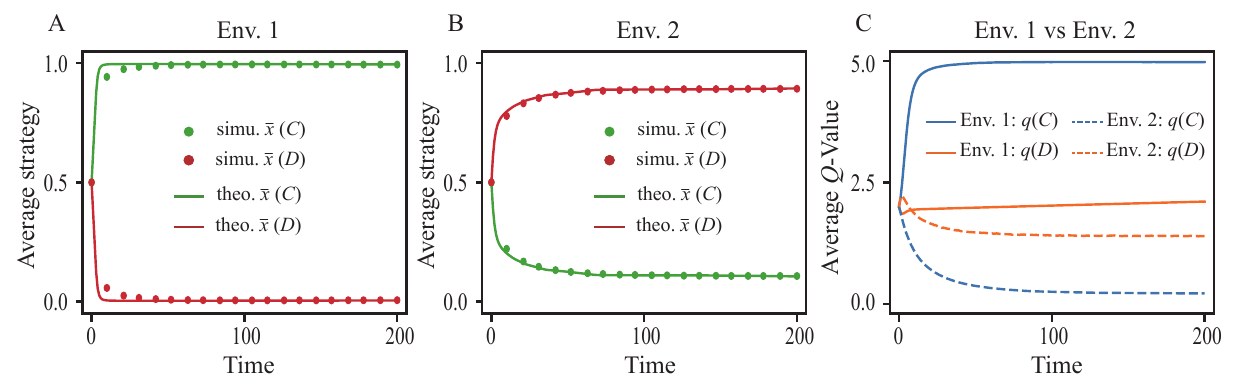}
	\caption{Population-level learning dynamics on a BA network under the two static payoff configurations. Panels A and B compare the population-average cooperation and defection probabilities obtained from Monte Carlo (MC) simulations and the mean-field (MF) model under Env. 1 and Env. 2, respectively. Panel C shows the MC population-average values of $q(C)$ and $q(D)$, with solid and dashed curves corresponding to Env. 1 and Env. 2, respectively.}
	\label{fig:static-games}
\end{figure*}

Having established the accuracy of the mean-field approximation across representative network structures and finite population sizes, we next examine the learning dynamics under different payoff configurations and environmental parameters.

Env. 1 and Env. 2 provide two static payoff baselines. Because $R^{(0)}=R^{(1)}$ in both configurations, the payoff matrix is independent of the environmental state. The feedback channel from the environment to the learning dynamics is therefore removed, although the population strategy may still affect the environmental state. These two configurations isolate the learning behavior under a fixed Harmony game and a fixed Prisoner's Dilemma, respectively.

As shown in Fig.~\ref{fig:static-games}A, cooperation is strictly dominant under Env. 1, and the population-average cooperation probability rapidly approaches a value close to one. Under Env. 2, defection is strictly dominant, and Fig.~\ref{fig:static-games}B shows that the population evolves toward a defection-dominated state while retaining a small but nonzero probability of cooperation. In both cases, the mean-field predictions closely track the Monte Carlo results during the initial change and at the subsequent stationary level.

Panel C explains why the outcomes under the two static games are not exact mirror images. Under Env. 1, the learned value of cooperation increases toward a substantially higher level than that of defection. Under Env. 2, the learned value of defection remains higher than that of cooperation, but the difference between the two values is smaller. Because a finite inverse temperature $\beta$ assigns a positive probability to every action, the system does not reach the absorbing states of exact full cooperation or exact full defection. Instead, the difference between $q(C)$ and $q(D)$ determines the strength of the softmax preference for the dominant action.

\begin{figure*}[t]
	\centering
	\includegraphics[width=\textwidth]{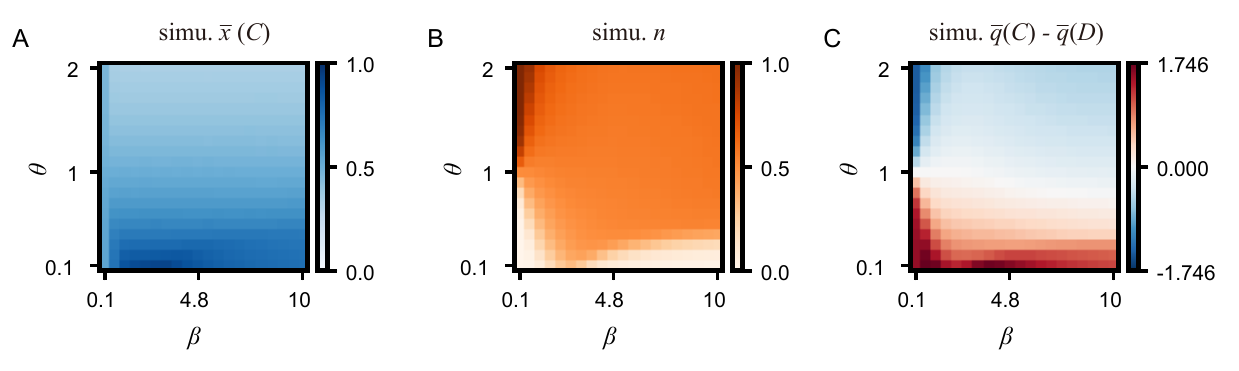}
	\caption{Late-time Monte Carlo estimates under Env. 3 in the $(\beta,\theta)$ parameter plane. Panels A--C show the population-average cooperation probability $\bar{x}(C)$, the environmental state $n$, and the population-average action-value difference $\bar{q}(C)-\bar{q}(D)$, respectively.}
	\label{fig:parameter-plane}
\end{figure*}
To examine how policy selectivity and environmental feedback jointly affect the coupled dynamics, we vary the inverse temperature over $\beta\in[0.1,10]$ and the feedback parameter over $\theta\in[0.1,2]$. Fig.~\ref{fig:parameter-plane} reports the resulting late-time averages under Env. 3. The three panels describe the behavioral, environmental, and action-value responses to the same parameter combinations.

As shown in Fig.~\ref{fig:parameter-plane}A, the cooperation probability is close to $0.5$ when $\beta$ is small because the softmax policy assigns similar probabilities to the two actions. As $\beta$ increases, the dependence on $\theta$ becomes more pronounced: cooperation remains high at small $\theta$ but decreases as $\theta$ increases. Under Env. 3, a larger $\theta$ allows cooperation to improve the environmental state more readily. The resulting increase in $n$, however, shifts the payoff matrix from the Harmony game toward the Prisoner's Dilemma, thereby weakening the learned incentive for cooperation.

Panel B shows that the environmental state generally increases with $\theta$. At larger $\theta$, the environmental state remains high across most of the investigated $\beta$ range because the cooperation level required for positive environmental growth, $1/(1+\theta)$, is lower. At smaller $\theta$, the environmental outcome depends non-monotonically on $\beta$: intermediate policy selectivity produces a higher environmental level than either nearly random action selection or strongly selective behavior. 

As shown in Fig.~\ref{fig:parameter-plane}C, the population-average action-value difference is positive throughout most of the low-$\theta$ region, indicating a higher learned value for cooperation. As $\theta$ increases, $\bar{q}(C)-\bar{q}(D)$ decreases and changes sign near the central part of the investigated $\theta$ range. The nearly horizontal zero-value boundary indicates that $\theta$ primarily governs this reversal in action-value ordering, whereas $\beta$ modulates the magnitude of the value difference on either side of the boundary. 

\begin{figure*}
	\centering
	\includegraphics[width=1\textwidth]{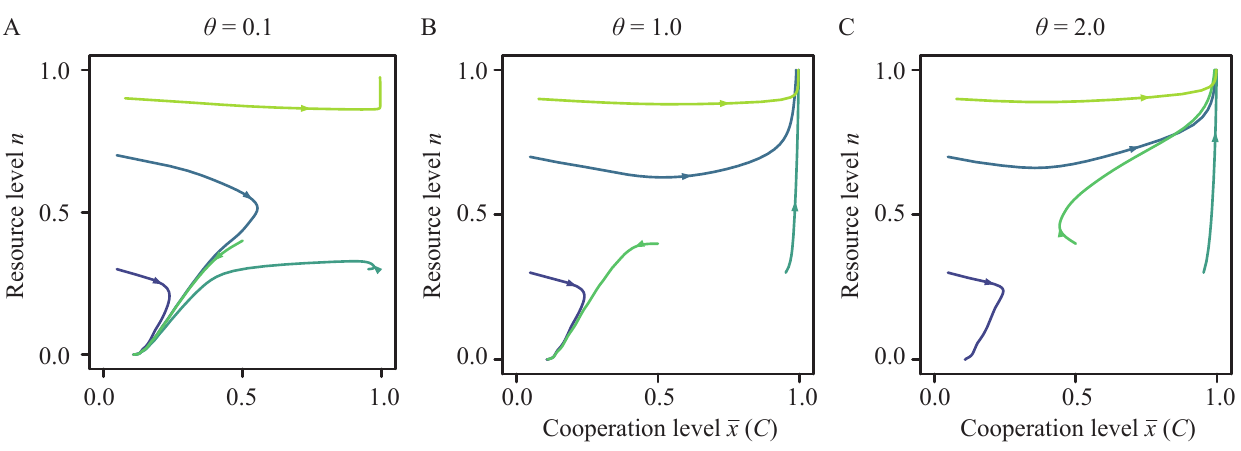}
	\caption{Projected mean-field trajectories under reinforcing payoff feedback (Env.~4) in the population-average cooperation--environment plane $(\bar{x}(C),n)$. Panels A--C correspond to $\theta=0.1$, $1.0$, and $2.0$, respectively, and arrows indicate the direction of increasing time.}
	\label{fig:initial-condition-trajectories}
\end{figure*}

To examine the sensitivity of the coupled dynamics to the initial learning bias and resource level, Fig.~\ref{fig:initial-condition-trajectories} presents projections of the mean-field trajectories onto the $(\bar{x}(C),n)$ plane under Env.~4. The trajectories exhibit pronounced initial-condition dependence. Initial states with both low cooperation and limited environmental resources evolve toward a resource-depleted, low-cooperation regime. By contrast, sufficiently favorable initial learning biases or resource levels can direct the system toward a resource-rich regime with a high cooperation probability. This separation arises from the reinforcing feedback in Env.~4: a favorable environment increases the relative payoff of cooperation, while increased cooperation subsequently promotes environmental recovery.

The influence of $\theta$ can be understood from the environmental growth condition
\begin{equation*}
	\bar{x}(C,t)>\frac{1}{1+\theta}.
\end{equation*}
The corresponding cooperation thresholds are approximately $0.909$, $0.5$, and $0.333$ for Panels A--C, respectively. Consequently, increasing $\theta$ allows a broader set of trajectories to enter the region of positive environmental growth. Nevertheless, the low-cooperation, low-resource trajectory still approaches environmental depletion even at $\theta=2$, showing that stronger feedback does not necessarily overcome an unfavorable early-stage evolution.

\begin{figure*}
	\centering
	\includegraphics[width=1\textwidth]{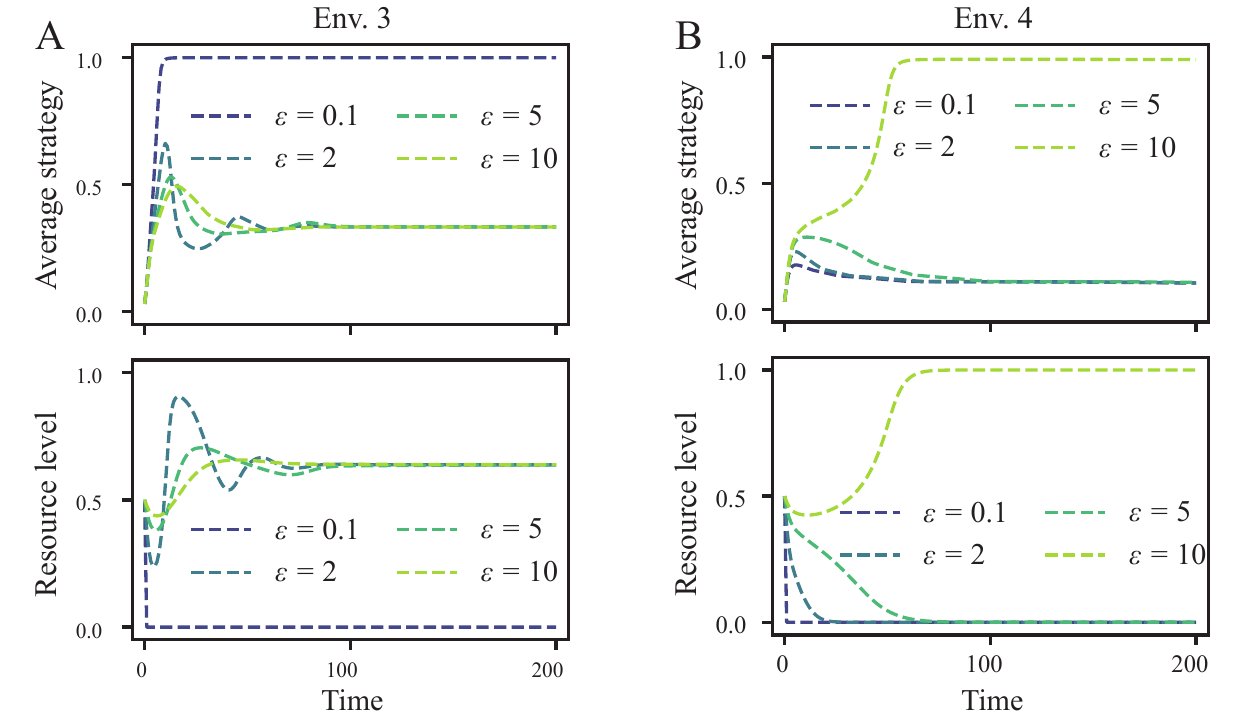}
	\caption{Mean-field transport predictions under different environmental response timescales. Panel A corresponds to Env.~3, and Panel B corresponds to Env.~4. In each panel, the upper and lower subplots show the population-average cooperation level $\bar{x}(C,t)$ and the environmental state $n(t)$, respectively.}
	\label{fig:timescale-effects}
\end{figure*}

We next examine the impact of the environmental response timescale $\epsilon$ on the system's evolutionary behavior (Fig.~\ref{fig:timescale-effects}). In all calculations, the initial cooperation level is set to a low value to emphasize the role of the environmental response during the early transient stage.

Under the Env.~3 payoff structure, a relatively small $\epsilon$ produces a rapid environmental response to the initially low cooperation level. Specifically, for $\epsilon=0.1$, the resource loss during the first learning round is sufficient for the projected environmental update to reach $n_1=0$ before the learning dynamics can substantially adjust. Because $n=0$ is an absorbing boundary of the projected update, the resource remains depleted thereafter, while cooperation subsequently increases under the cooperation-favorable payoff configuration. For $\epsilon=2,5,$ and $10$, the slower environmental response allows the learning and environmental processes to interact over a longer period, producing transient excursions before the system approaches an intermediate long-time state (Fig.~\ref{fig:timescale-effects}A).

Under the Env.~4 payoff structure, a small $\epsilon$ similarly causes the resource state to decline rapidly toward $n=0$. For $\epsilon=0.1,2,$ and $5$, the resulting depleted environment corresponds to a defection-favorable payoff configuration, and the population ultimately remains at a low cooperation level (Fig.~\ref{fig:timescale-effects}B). In contrast, when $\epsilon=10$, the slower environmental response provides a sufficient time window for cooperation to increase before the resource state reaches the depleted boundary. The resulting positive feedback drives the environmental state toward its upper boundary and leads to a highly cooperative long-time state. These results demonstrate that a small $\epsilon$ can cause the environment to reach and remain at the depleted boundary, whereas a sufficiently large $\epsilon$ allows the learning and environmental processes to remain coupled and can lead to a different outcome.

\section{Conclusion}

In this paper, we developed a coupled learning--environment model that combines stateless multi-agent $Q$-learning on fixed graphs with endogenous environmental feedback. The model describes how agents update action values from neighbor-averaged rewards, how the population-average strategy changes the environmental state, and how the resulting environmental state modifies future payoffs. Under a first-order mean-field closure, we derived a deterministic transport equation for the population distribution of $Q$-values and coupled it with a projected environmental update defined on the discrete learning timescale.

The comparison with finite-network Monte Carlo simulations on RG, ER, BA, and RGG networks shows that the mean-field system captures the main macroscopic cooperation and environmental trajectories over the tested parameter ranges. The trajectory-level RMSE is generally of order $10^{-2}$ and decreases as the population size and average degree increase. This trend indicates that finite-population fluctuations and limited neighborhood sampling are important sources of the deviation between the finite-network dynamics and the population-level approximation.

The results also clarify how environmental feedback shapes collective learning. Under the compensatory payoff configuration, the feedback parameter changes the learned action-value ordering and the resulting cooperation level. Under the reinforcing payoff configuration, different initial learning biases and environmental states can lead to distinct long-run outcomes, demonstrating pronounced path dependence. The environmental timescale provides an additional regulating factor: when the response is too rapid, the resource state can reach the depleted boundary before the learning process adapts, whereas a slower response allows behavioral adjustment and environmental recovery to remain coupled.

The present framework relies on a mean-field closure in which neighboring actions are approximated by population-level strategy statistics. Accordingly, its predictive accuracy may decrease in sparse networks with strong local correlations, community structure, or assortative mixing. The numerical analysis in this paper focuses on fixed networks, binary-action games, and a linear environmental feedback function. These choices specify the benchmark scenarios considered in the numerical experiments, rather than fundamental restrictions of the theoretical framework. The formulation is written for a general finite action set, and the transport equation describes the deterministic mean-drift dynamics of the learning process. More general payoff--environment mappings and nonlinear feedback functions can therefore be incorporated by replacing the corresponding model terms. Future work will investigate these extensions together with state-dependent learning and adaptive or co-evolving networks.


\bibliographystyle{IEEEtran}
\bibliography{references}

\newpage

 




\vfill

\end{document}